\documentclass[a4paper,UKenglish,cleveref, autoref, thm-restate]{lipics-v2019}

\usepackage{hyperref}
\usepackage{amssymb, amsmath}
\usepackage{amsthm}
\usepackage{dsfont}
\usepackage[ruled,vlined,linesnumbered]{algorithm2e}
\usepackage{tikz}
\usetikzlibrary{shapes,backgrounds}

\newcommand{\val}{\mbox{Val}}
\newcommand{\M}{\text{\sc max}}
\newcommand{\m}{\text{\sc min}}
\newcommand\comments[1]{}

\everymath{\displaystyle}

\DeclareMathOperator*{\argmax}{arg\,max}

\title{A Recursive Algorithm for Solving Simple Stochastic Games} 

\author{Xavier Badin de Montjoye}{Université Paris Saclay, UVSQ, DAVID}{xavier.badin-de-montjoye2@uvsq.fr}{}{}

\authorrunning{X. Badin De Montjoye} 

\Copyright{Xavier Badin de Montjoye} 

\ccsdesc[500]{Theory of computation~Algorithmic game theory}

\keywords{Simple Stochastic Games, Strategy Improvement, Parametrized Complexity, Recursif algorithm} 

\acknowledgements{We would like to thanks David Auger, Pierre Coucheney and Yann Strozecki for sharing their insights on SSGs.}

\nolinenumbers 

\EventEditors{John Q. Open and Joan R. Access}
\EventNoEds{2}
\EventLongTitle{42nd Conference on Very Important Topics (CVIT 2016)}
\EventShortTitle{CVIT 2016}
\EventAcronym{CVIT}
\EventYear{2016}
\EventDate{December 24--27, 2016}
\EventLocation{Little Whinging, United Kingdom}
\EventLogo{}
\SeriesVolume{42}
\ArticleNo{23}

\begin{document}

\maketitle

\begin{abstract}
    We present two recursive strategy improvement algorithms for solving simple stochastic games. First we present an algorithm for solving SSGs of degree $d$ that uses at most $O\left(\left\lfloor(d+1)^2/2\right\rfloor^{n/2}\right)$ iterations, with $n$ the number of MAX vertices. Then, we focus on binary SSG and propose an algorithm that has complexity $O\left(\varphi^nPoly(N)\right)$ where $\varphi = (1 + \sqrt{5})/2$ is the golden ratio. To the best of our knowledge, this is the first deterministic strategy improvement algorithm that visits $2^{cn}$ strategies with $c < 1$.
\end{abstract}


\section{Introduction}

Simple stochastic games (SSG) are a restriction  introduced by Condon~\cite{condon1990algorithms,condon1992complexity} of the notion of stochastic games defined by Shapley~\cite{Shapley1095}. An SSG is a turn-based zero-sum game with perfect information played by two players named \M\ and \m. A token is placed on a directed graph and moves alongside the graph arcs. The set of vertices is partitioned in \M, \m, Random and Sink vertices. If the token is in a \M\ or \m\ vertices, the corresponding player chooses its next position in the outneighbourhood of the currently occupied vertex. In the case of random vertices, the token is moved randomly. Finally, when a sink vertex $s$ is reached, the game ends and player \m\ must pay some penalty value $\val(s)$ to player \M. The goal of \m\ is to minimise the expected penalty, and the goal of \M\ is to maximise it. One reason to study SSGs and their related complexity is that the stochastic versions of classical games such as parity games, mean and discounted payoff are all equivalent to SSG~\cite{Anderson2009}.

Simple stochastic games admit a pair of optimal strategies~\cite{condon1992complexity} whose expected value satisfies a Nash equilibrium. Our goal is to compute such a pair of optimal strategy. This problem is known to be in $\textsc{PPAD}$~\cite{juba2005hardness} a subset of the class $\textsc{FNP}$, but is not known to be in $\textsc{FP}$. SSG have several applications, such as modelling autonomous urban driving~\cite{Chen2013urbandriving} or in the domain of model checking of modal $\mu$-calculus~\cite{stirling1999bisimulation}.

We note $n$ the number of \M\ vertices, $N$ the number of total vertices and $r$ the number of random vertices. Every known algorithm in the literature for solving SSGs has exponential time complexity bounds. There exist several FPT algorithms that can be used for specific families of SSGs. For instance, for SSGs with random vertices of degree $2$ and uniform probability distribution, Ibsen-Jensen and Milterson gives a value iteration algorithm in~\cite{ibsen2012solving} with complexity $O\left(r2^{r}(r\log(r)+N)\right)$.

One of the main families of algorithms for solving SSGs are strategy improvement algorithms. In the case of SSG with \M\ vertices of outdegree exactly $2$, Tripathi, Valkanova and Kumar present in~\cite{tripathi2011strategy} an algorithm with $O\left(2^{n}/n\right)$ iterations, which is the only current bounds that deterministically improve the trivial $2^n$ iterations of checking every possible strategies of player \M. Ludwig offers a randomised algorithm in~\cite{ludwig1995subexponential}, which does $\displaystyle 2^{O\left(\sqrt{n}\right)}$ iterations on average. In this paper, we provide a deterministic algorithm for the same family of SSGs that has complexity $O\left(2^{cn}Poly(N)\right)$ for some $c<1$. Moreover, we will provide the first deterministic algorithm of this family with parametrised complexity in the degree $d$ of \M\ vertices and number of \M\ vertices, improving the trivial bound of $O\left(d^nPoly(N)\right)$ to $O\left(((d+1)/\sqrt{2})^{n}Poly(N)\right)$.

In the general case, Gimbert and Horn give in~\cite{gimbert2008simple} a strategy improvement algorithm whose complexity is a function of the random vertices with an algorithm of complexity $\displaystyle O\left(r!Poly(N)\right)$. Moreover, in~\cite{auger2019solving}, Auger Coucheney and Strozecki present a stochastic algorithm that runs in $2^{O(r)}$.

Auger, Badin de Montjoye and Strozecki present in~\cite{auger2021} a general formulation for strategy improvement algorithms that offers a general bound on the complexity of all such algorithms depending on the format of the probability value of the random nodes. It states that if there is some $q$ such that all probabilities are of the format $p/q$ and if each iteration of the algorithm is done in polynomial time, then any strategy improvement algorithm runs in $O\left(nq^rPoly(N)\right)$.

\subsection*{Contributions}

In this paper, we focus on SSGs whose \M\ vertices have outdegree $d$ with no constraints on the probability distribution of random vertices. We introduce two new recursive algorithms to solve SSGs. The first one, in Section~\ref{sec:rec2}, fixes the strategies on two vertices and recursively solves the rest of the SSG. This algorithm has bound $O\left(\left(\left\lfloor\frac{(d+1)^2}{2}\right\rfloor-1\right)^{n/2}Poly(|G|)\right)$. The second algorithm presented in Section~\ref{sec:recT} works only for SSG of degree $2$. However, it achieves a better bound that the $O\left(\sqrt{3}^nPoly(N)\right)$ of the first algorithm by reaching $O\left(\varphi^nPoly(N)\right)$ with $\varphi= (1+\sqrt{5})/2$ the golden ratio, which is, to the best of our knowledge, the best complexity for deterministic algorithm on this family of SSGs. Moreover, our algorithm does not require SSGs to be stopping, a common technical hypothesis, which may require a squarring of the number of vertices to be met.

\section{An Overview of Simple Stochastic Games}

    Simple Stochastic Games where introduce by Anne Condon in~\cite{condon1990algorithms}. We give a definition close to the one given in~\cite{tripathi2011strategy,auger2021}.
    
    \begin{definition}
        A Simple Stochastic Game (SSG) is a directed graph $G=(V,E)$ with a partition of the vertex set $V$ in $V_{\M}$, $V_{\m}$, $V_{R}$ and $V_{S}$ respectively called, \M, \m, random and sink vertices such that:
        \begin{itemize}
            \item every vertex of $V_{\M}$ and $V_{\m}$ has outdegree at least two.
            \item every vertex $x$ of $V_{R}$ has outdegree at least one, and an associated rational probability distribution $p_x(\cdot)$ on the outneighbourhood of $x$.
            \item every vertex $x \in V_S$, there is an associated rational value $\val(x)$ in the closed interval $[0,1]$.
        \end{itemize}
    \end{definition}
    
    \begin{definition}
        A binary SSG is an SSG where every vertex of $V_{\M}$ has outdegree $2$. An SSG is of degree $d$ if its \M\ vertices are of degree at most $d$.
    \end{definition}
    
    In this article we denote $|V_{\M}|$ by $n$. We will write $|G|$ the size of the representations of the game $G$ in bits. We present an instance of an SSG in Figure~\ref{fig:def_SSG}.
    
    \begin{figure}
\centering

    \begin{tikzpicture}
    
        \node[draw, circle] (x1) at (0,4) {$x_1$};
		\node[draw, circle, fill=gray!50] (1) at (0.5,1) {$1$};
		\node[draw, rectangle, minimum height=0.8cm, minimum width=0.8cm] (n1) at (2.5,2) {$n_1$};
		\node[draw, diamond, aspect=1] (r1) at (1.5,3) {$r_1$};
		\node[draw, rectangle, minimum height=0.8cm, minimum width=0.8cm] (n2) at (1.5,5.5) {$n_2$};
		\node[draw, circle] (x2) at (2.5,4) {$x_2$};
		\node[draw, rectangle, minimum height=0.8cm, minimum width=0.8cm] (n3) at (3,7) {$n_3$};
		\node[draw, diamond, aspect=1] (r2) at (4,2.5) {$r_2$};
		\node[draw, circle] (x3) at (4,5.5) {$x_3$};
		\node[draw, circle] (x4) at (5.5,4) {$x_4$};
		\node[draw, diamond, aspect=1] (r3) at (6,1) {$r_3$};
		\node[draw, circle, fill=gray!50] (0) at (6,7) {$0$};

		\draw[->,>=latex] (x1)--(1);
		\draw[->,>=latex] (x1) to[bend right = 20] (n2);
		\draw[->,>=latex] (n1)--(1);
		\draw[->,>=latex] (n1)--(r1);
		\draw[->,>=latex] (r1)--(1);
		\draw[->,>=latex] (r1)--(r2);
		\draw[->,>=latex] (r1)--(x2);
		\draw[->,>=latex] (n2) to[bend right = 20] (x1);
		\draw[->,>=latex] (n2)--(n3);
		\draw[->,>=latex] (x2)--(n2);
		\draw[->,>=latex] (x2)--(r2);
		\draw[->,>=latex] (n3)--(0);
		\draw[->,>=latex] (n3)--(x3);
		\draw[->,>=latex] (r2)--(n1);
		\draw[->,>=latex] (r2)--(x3);
		\draw[->,>=latex] (r2) to[bend right = 20] (x4);
		\draw[->,>=latex] (r2) to[bend right = 20] (r3);
		\draw[->,>=latex] (x3)--(0);
		\draw[->,>=latex] (x3)--(x2);
		\draw[->,>=latex] (r3)--(1);
		\draw[->,>=latex] (r3) to[bend right = 20] (r2);
		\draw[->,>=latex] (r3)--(x4);
		\draw[->,>=latex] (x4)--(x3);
		\draw[->,>=latex] (x4) to[bend right = 20] (r2);

    \end{tikzpicture}

\caption{Instance of an SSG where square, circle and diamond vertices are respectively \M, \m\ and random vertices and grey vertices are sinl vertices. The probability distributions of the random vertices is the uniform distribution over their children.}
\label{fig:def_SSG}
\end{figure}
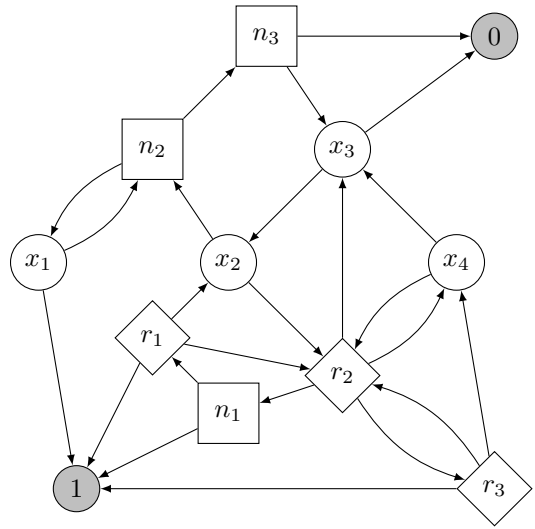
    
    The game is played by two players named \M\ and \m. The game starts by placing a token on some initial vertex $x_0$. Then, the token is moved according to the following rule. If the token is in a \M\ or a \m\ vertex $x$, then the corresponding player moves the token according to an outgoing arc from $x$. If the token is in a random vertex $x$, then the token is moved according to the probability distribution $p_x$. When the token reaches a sink $s$, then player \m\ has to pay player \M\ the value $\val(x)$. Informally, the goal of the game for \M\ is to maximise the value of the final sink and, conversely, the goal of the game for \m\ is to minimise it.
    
    This game is turn-based, with perfect information. The strategies consider by both players should thus be deterministic and only relies on the current position of the token. This is a well-known result on simple stochastic games and a proof of this can be found in~\cite{condon1990algorithms,tripathi2011strategy}. We thus only consider positional strategies.
    
    \begin{definition}
        A positional \M\ strategy is a function $\sigma$ from $V_{\M}$ to $V$ such that for all $x$ in $V_{\M}$, $(x,\sigma(x))$ is an arc of $G$.
        
        A positional \m\ strategy is a function $\tau$ from $V_{\m}$ to $V$ such that for all $x$ in $V_{\m}$, $(x,\tau(x))$ is an arc of $G$.
    \end{definition}
    
    In this paper when we talk about a pair of strategies $(\sigma,\tau)$, $\sigma$ is a positional \M\ strategy and $\tau$ is a positional \m\ strategy. We can now define the value vector of a pair of strategy.
    
    \begin{definition}
        For $(\sigma,\tau)$ a pair of strategies, and $x_0$ a vertex of $V$, the value $v_{\sigma,\tau}(x_0)$ is the expected gain for \M\ if both players play according to $\sigma$ and $\tau$. In other words:
        \[v_{\sigma,\tau}(x_0) = \sum\limits_{s \in V_S} \mathds{P}_{\sigma,\tau}(x_0 \longrightarrow s)\val(s)\]
        where $\mathds{P}_{\sigma,\tau}(x_0 \longrightarrow s)$ is the probability that the game ends in $s$ while starting in $x_0$ and such that when the token is in a \M\ vertex (resp. \m\ vertex) $x$ it moves to $\sigma(x)$ (resp. $\tau(x)$). 
    \end{definition}
    
    The value vector $v_{\sigma,\tau}$ is the vector $\left(v_{\sigma,\tau}(x)\right)_{x \in V}$. We compare value vectors according to the pointwise order. For two value vector $v$ and $v'$, $v > v'$ if, for all $x$, $v(x) \geq v'(x)$ and there is some $y \in V$ such that $v(y) > v'(y)$. As usual, $v \geq v'$ if $v > v'$ or if $v = v'$.
    
    Computing the value vector of a pair of strategies is equivalent to solving a Markov chain which can be done in time polynomial in $|G|$.
    
    For a \M\ strategy $\sigma$ we say that a \m\ strategy $\tau$ is a best response for $\sigma$ if and only if, for every \m\ strategy $\tau'$, we have $v_{\sigma,\tau} \leq v_{\sigma,\tau'}$. It is also well known that best response exists.
    
    \begin{proposition}[\cite{condon1990algorithms}]
        For every positional \M\ strategy $\sigma$, there exists a positional \m\ strategy $\tau$ that is a best response for $\sigma$. Moreover, a best response can be computed in polynomial time by linear programming.
    \end{proposition}
    
    In the same way, we can define the best response for a \m\ strategy.
    
    \begin{definition}
        For every \M\ strategy $\sigma$, we write $v_{\sigma} = v_{\sigma,\tau(\sigma)}$ where $\tau(\sigma)$ is a best response to $\sigma$.
    \end{definition}
    
    We say that $\sigma$ is better than $\sigma'$ or has greater value, or we note $\sigma > \sigma'$ if $v_{\sigma} > v_{\sigma'}$.
    
    It is well known in the literature (\cite{condon1990algorithms,tripathi2011strategy}) that there is a pair of positional strategies $(\sigma^{*},\tau^{*})$ that are called optimal strategies such that they are best response of each other. The value vector equilibrium $v_{\sigma^{*},\tau^{*}}$ is unique. Our goal is to compute this value vector. A known characteristic of a pair of optimal strategy is that its associated value vector satisfies local optimality.
    
    \begin{proposition}[\cite{condon1990algorithms}]\label{prop:local_optimality}
        Let $G=(V,E)$ be an SSG. For $\sigma$ a \M\ strategy, $(\sigma,\tau(\sigma))$ is a pair of optimal strategy if and only if $v_{\sigma}$ satisfies the following local optimality condition:
        \begin{itemize}
            \item for $x \in V_{\M}$, $v_{\sigma}(x) = \max \{v_{\sigma,\tau}(y)\;|\;(x,y) \in E\}$
            \item for $x \in V_{\m}$, $v_{\sigma}(x) = \min \{v_{\sigma,\tau}(y)\;|\;(x,y) \in E\}$
        \end{itemize}
    \end{proposition}
    
    It is important to notice that the converse of the proposition is true only because we consider the best response to a \M\ strategy. If we consider a \m\ strategy and its best response, then it does not hold anymore because of possible loops.

\section{Switch set and super-switch}

    \subsection{Switch set}
    
        In this section we present an important concept for strategy improvement algorithm: the switch set. We present several properties of the switch set for SSG that we use to prove the complexity of our algorithms. 
    
        \begin{definition}
            Let $G=(V,E)$ be an SSG and $\sigma$ a \M\ strategy. The switch set of $\sigma$, written $S_{\sigma}$, is the set of vertices $x$ such that there is $y$ with $(x,y) \in E$ and $v_{\sigma}(x) < v_{\sigma}(y)$
        \end{definition}
        
        \begin{definition}
            Let $\sigma$ be a \M\ strategy. For all $x \in S_{\sigma}$, the improvement set, written $IS_{\sigma}(x)$ is the set of neighbours $y$ of $x$ such that $v_{\sigma}(y) > v_{\sigma}(x)$ and the best improvement option is defined as $bio_{\sigma}(x) = \argmax\limits_{y \in IS_{\sigma(x)}} \{v_{\sigma}(y)\}$.
        \end{definition}
        
        In other words, the switch set of a \M\ strategy is the set of \M\ vertices that do not satisfy the local optimality condition presented Proposition~\ref{prop:local_optimality}. This notion directly gives the concept of $\sigma$-switch.
        
        \begin{definition}
            Let $\sigma$ be a \M\ strategy with a non-empty switch set $S_\sigma$. A \M\ strategy $\sigma'$ is said to be a $\sigma$-switch if $\sigma' \neq \sigma$, for all $x \in V_{\M} \smallsetminus S_{\sigma}$, $\sigma(x) = \sigma'(x)$ and for all $x \in S_{\sigma}$ such that $\sigma'(x) \neq \sigma(x)$, $\sigma'(x) \in IS_{\sigma}(x)$.
        \end{definition}
        
        A $\sigma$-switch is a strategy, where the strategy on vertices that satisfy local optimality for $\sigma$ has been kept and it had been changed on some vertices that did not satisfy local optimality. Informally, as the name implies, it is a strategy where we "switch" the strategy on some vertices that can achieve immediate better value by selecting another child. A representation of a switch is given in Figure~\ref{fig:def_ss}.
        
        \begin{definition}
            For $\sigma$ a \M\ strategy, the total switch $\bar{\sigma}$ is the $\sigma$-switch where for all $x \in S_{\sigma}$, $\bar{\sigma}(x) = bio_{\sigma}(x)$.
        \end{definition}
        
        \begin{proposition}[\cite{condon1990algorithms,tripathi2011strategy}]\label{prop:switch}
        
            For $\sigma$ a \M\ strategy and $\sigma'$ a $\sigma$-switch, $v_{\sigma'} > v_{\sigma}$.
        
        \end{proposition}
        
        The demonstration of this proposition uses the condition that the SSG is stopping, which means that the SSG ends in a sink with probability $1$ and this for any pair of strategy. However, it had been shown that this condition is unnecessary~\cite{CHATTERJEE2013reachability,auger2021}. 
        
        \begin{corollary}\label{coro:empty_set_optimal}
            A positional \M\ strategy $\sigma$ is optimal if and only if $S_{\sigma}$ is empty.
        \end{corollary}
        
        This property directly gives a family of algorithms called Hoffmann-Karp algorithms. Starting from some \M\ strategy $\sigma$, compute $v_{\sigma}$ then, if $S_{\sigma}$ is not empty, choose a $\sigma$-switch and iterate. Since the number of strategies is bounded by $d^{n}$ where $d$ is the degree of $G$, this algorithm terminates and provides an optimal \M\ strategy. Tripathi, Valkanova and Kumar shows in \cite{tripathi2011strategy} that for binary SSG, if at each iteration, we consider the $\sigma$-switch $\bar{\sigma}$, then the algorithm needs $O\left(2^{n}/n\right)$ iterations. We will show in Theorem~\ref{thm:tripathi_super_switch} an extended version of their main result that we use to find a bound of the Algorithm~\ref{alg:RP} presented in Section~\ref{sec:rec2}.
        
        For $T$ a set of \M\ vertices, and $\sigma$ a \M\ strategy, we define the subgame $G_{|T[\sigma]}$ as the game $G$ where some \M\ vertices have been replaced with random vertices that go to some vertex with probability one according to the \M\ strategy $\sigma$. In other words, $G_{|T[\sigma]}$ is the game $G$ where \M\ has to play as in $\sigma$ in $T$. 
        
        \begin{definition}
            Let $G = (V,E)$ an SSG. For $\sigma$ a strategy and $T$ a subset of $V_{\M}$ we write $G_{|T[\sigma]}=(V,E')$ the SSG that is a copy of $G$ where $E' = E \smallsetminus \{(x,y)\;|\;x \in T,\;y \neq \sigma(x)\}$ and all vertices $x$ of $T$ are random vertices, with associated probability distribution $p_x(\sigma(x)) = 1$.
        \end{definition}
        
        We provide an example of this transformation in Figure~\ref{fig:sub_game}.
        
        \begin{figure}
\centering

    \begin{tikzpicture}
    
        \node[draw, circle] (x1) at (0,4) {$x_1$};
		\node[draw, circle, fill=gray!50] (1) at (0.5,1) {$1$};
		\node[draw, rectangle, minimum height=0.8cm, minimum width=0.8cm] (n1) at (2.5,2) {$n_1$};
		\node[draw, diamond, aspect=1] (r1) at (1.5,3) {$r_1$};
		\node[draw, rectangle, minimum height=0.8cm, minimum width=0.8cm] (n2) at (1.5,5.5) {$n_2$};
		\node[draw, circle] (x2) at (2.5,4) {$x_2$};
		\node[draw, rectangle, minimum height=0.8cm, minimum width=0.8cm] (n3) at (3,7) {$n_3$};
		\node[draw, diamond, aspect=1] (r2) at (4,2.5) {$r_2$};
		\node[draw, circle] (x3) at (4,5.5) {$x_3$};
		\node[draw, circle] (x4) at (5.5,4) {$x_4$};
		\node[draw, diamond, aspect=1] (r3) at (5,1) {$r_3$};
		\node[draw, circle, fill=gray!50] (0) at (5,7) {$0$};
		
		\node (tr) at (6.5,3.5) {$\longrightarrow$};
		
		\node[draw, diamond, aspect=1] (x1b) at (7.5,4) {$r'_1$};
		\node[draw, circle, fill=gray!50] (1b) at (8,1) {$1$};
		\node[draw, rectangle, minimum height=0.8cm, minimum width=0.8cm] (n1b) at (10,2) {$n_1$};
		\node[draw, diamond, aspect=1] (r1b) at (9,3) {$r_1$};
		\node[draw, rectangle, minimum height=0.8cm, minimum width=0.8cm] (n2b) at (9,5.5) {$n_2$};
		\node[draw, circle] (x2b) at (10,4) {$x_2$};
		\node[draw, rectangle, minimum height=0.8cm, minimum width=0.8cm] (n3b) at (10.5,7) {$n_3$};
		\node[draw, diamond, aspect=1] (r2b) at (11.5,2.5) {$r_2$};
		\node[draw, diamond, aspect=1] (x3b) at (11.5,5.5) {$r'_3$};
		\node[draw, diamond, aspect=1] (x4b) at (13,4) {$r'_4$};
		\node[draw, diamond, aspect=1] (r3b) at (12.5,1) {$r_3$};
		\node[draw, circle, fill=gray!50] (0b) at (12.5,7) {$0$};
		
		\draw[->,>=latex] (x1)--(1);
		\draw[->,>=latex] (x1) to[bend right = 20] (n2);
		\draw[->,>=latex] (n1)--(1);
		\draw[->,>=latex] (n1)--(r1);
		\draw[->,>=latex] (r1)--(1);
		\draw[->,>=latex] (r1)--(r2);
		\draw[->,>=latex] (r1)--(x2);
		\draw[->,>=latex] (n2) to[bend right = 20] (x1);
		\draw[->,>=latex] (n2)--(n3);
		\draw[->,>=latex] (x2)--(n2);
		\draw[->,>=latex] (x2)--(r2);
		\draw[->,>=latex] (n3)--(0);
		\draw[->,>=latex] (n3)--(x3);
		\draw[->,>=latex] (r2)--(n1);
		\draw[->,>=latex] (r2)--(x3);
		\draw[->,>=latex] (r2) to[bend right = 20] (x4);
		\draw[->,>=latex] (r2) to[bend right = 20] (r3);
		\draw[->,>=latex] (x3)--(0);
		\draw[->,>=latex] (x3)--(x2);
		\draw[->,>=latex] (r3)--(1);
		\draw[->,>=latex] (r3) to[bend right = 20] (r2);
		\draw[->,>=latex] (r3)--(x4);
		\draw[->,>=latex] (x4)--(x3);
		\draw[->,>=latex] (x4) to[bend right = 20] (r2);
		
		\draw[->,>=latex] (x1b)--(1b);
		\draw[->,>=latex] (n1b)--(1b);
		\draw[->,>=latex] (n1b)--(r1b);
		\draw[->,>=latex] (r1b)--(1b);
		\draw[->,>=latex] (r1b)--(r2b);
		\draw[->,>=latex] (r1b)--(x2b);
		\draw[->,>=latex] (n2b) to[bend right = 20] (x1b);
		\draw[->,>=latex] (n2b)--(n3b);
		\draw[->,>=latex] (x2b)--(n2b);
		\draw[->,>=latex] (x2b)--(r2b);
		\draw[->,>=latex] (n3b)--(0b);
		\draw[->,>=latex] (n3b)--(x3b);
		\draw[->,>=latex] (r2b)--(n1b);
		\draw[->,>=latex] (r2b)--(x3b);
		\draw[->,>=latex] (r2b) to[bend right = 20] (x4b);
		\draw[->,>=latex] (r2b) to[bend right = 20] (r3b);
		\draw[->,>=latex] (x3b)--(x2b);
		\draw[->,>=latex] (r3b)--(1b);
		\draw[->,>=latex] (r3b) to[bend right = 20] (r2b);
		\draw[->,>=latex] (r3b)--(x4b);
		\draw[->,>=latex] (x4b) to[bend right = 20] (r2b);

    \end{tikzpicture}

\caption{Transformation of the game $G$ in the game $G_{T[\sigma]}$ where $T=\{x_1;x_3,x_4\}$ and $\sigma(x_1) = 1$, $\sigma(x_3) = x_2$ and $\sigma(x_4) = r2$. The probability distribution on the random vertices is the uniform distribution.}
\label{fig:sub_game}
\end{figure}
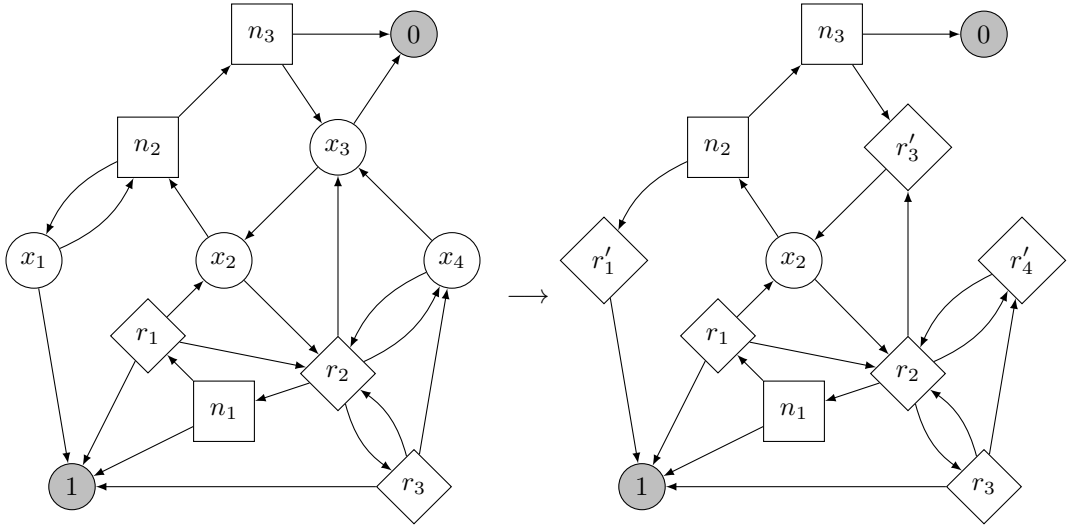
        
        There is a bijection between the strategy of $G_{|T[\sigma]}$ and the strategy of $G$ that plays has in $\sigma$ in T. In the rest of the paper, we identify those two sets of strategy. Moreover, for a strategy $\sigma'$ that plays as in $\sigma$ in $T$ and any \m\ strategy $\tau$, there is equality between the value vectors $v_{\sigma',\tau}$ in $G$ and $G_{|T[\sigma]}$.
        
        \begin{lemma}[\cite{tripathi2011strategy}]\label{lem:opt_switch}
            For any \M\ strategy $\sigma$, $\sigma$ is optimal in $G_{|S_{\sigma}[\sigma]}$.
        \end{lemma}
        
        \begin{proof}
            The switch set of $\sigma$ in $G_{|T[\sigma]}$ is empty, hence, by Corollary~\ref{coro:empty_set_optimal}, $\sigma$ is optimal in $G_{|T[\sigma]}$.
        \end{proof}

        The following proposition also appears in \cite{tripathi2011strategy} and it allows us to study strategies by just looking at their switch set.
        
        \begin{proposition}\label{prop:set_inclusion}
            Let $G$ be a binary SSG. For $\sigma$ and $\sigma'$ two \M\ strategies such that $S_{\sigma} \subsetneq S_{\sigma'}$, then $v_{\sigma} > v_{\sigma'}$.
        \end{proposition}

        \begin{proof}
            Let $\sigma$ and $\sigma'$ be two \M\ strategies such that $S_{\sigma} \subsetneq S_{\sigma'}$. If $\sigma'$ is a strategy of $G_{|S_{\sigma}[\sigma]}$, then $\sigma'$ is not optimal in $G_{|S_{\sigma}[\sigma]}$ by Corollary~\ref{coro:empty_set_optimal} and by Lemma~\ref{lem:opt_switch}, $v_{\sigma}>v_{\sigma'}$. Otherwise, we consider $\sigma''$ that plays as in $\sigma'$ in $V_{\M} \smallsetminus S_{\sigma}$ and as in $\sigma$ in $S_{\sigma}$. $\sigma''$ is a $\sigma'$-switch, and a strategy of $G_{|S_{\sigma}[\sigma]}$, thus: $v_{\sigma'} < v_{\sigma''} \leq v_{\sigma}$
        \end{proof}
        
        In the same way, we can also consider the case where both switch sets are equal.
        
        \begin{lemma}\label{lem:eq_vv}
            Let $G$ be a binary SSG. For $\sigma$ and $\sigma'$ two \M\ strategies such that $S_{\sigma} = S_{\sigma'}$, then $v_{\sigma} = v_{\sigma'}$ and for every $x$ in $S_{\sigma}$, $\sigma(x) = \sigma'(x)$
        \end{lemma}
        
        \begin{proof}
            Let $\sigma$ and $\sigma'$ be two strategies that satisfies our condition, and for the sake of contradiction, let us assume that there is $x$ in $S_{\sigma}$ such that $\sigma(x) \neq \sigma'(x)$. Then, the strategy $\sigma''$ that plays as $\sigma$ in $S_{\sigma}$ and as $\sigma'$ in the other vertices is a $\sigma'$-switch. Thus $v_{\sigma''} > v_{\sigma'}$. Since $\sigma''$ plays as $\sigma$ on $S_{\sigma}$, by Lemma~\ref{lem:opt_switch}, $v_{\sigma} \geq v_{\sigma''}$. Thus $v_{\sigma} > v_{\sigma'}$. We similarly can prove that $v_{\sigma'} > v_{\sigma}$, which yields a contradiction. Hence, $\sigma$ and $\sigma'$ plays similarly on $S_{\sigma}$. The strategy $\sigma$ and $\sigma'$ are optimal strategy of $G_{|S_{\sigma}[\sigma]}$ and thus have the same value.
        \end{proof}
        
    \subsection{Super-Switch}
    
        In this section, we will extend the classic notion of switch in order to add a perfect resolution on some of its vertices.
        
        A super-switch is obtained by switching some vertices of $\sigma$, fixing the strategy on a set of \M\ vertices that include the switched vertices and then considering the optimal strategy of the subgame.
        
        Let $\sigma$ be a \M\ strategy with a non-empty switch set $S_{\sigma}$ and $T \subset V_{\M}$ with $T \cap S_{\sigma} \neq \emptyset$. A $(\sigma,T)$-super-switch is a strategy $\tilde{\sigma}$ obtained from an intermediate $\sigma$-switch $\sigma'$ such that $\forall x \notin T,\; \sigma(x) =\sigma'(x)$ by computing an optimal strategy of $G_{|T[\sigma']}$.
        
        We gives a representation of a super-switch in Figure~\ref{fig:def_ss}
        
        \begin{figure}
\centering

    \begin{tikzpicture}
    
        \node[draw, circle, fill=gray!50] (0) at (0,4) {$0$};
        \node[draw, diamond, aspect=1] (r1) at (0,2) {$r_1$};
        \node (vr1) at (-0.5,2.5) {$2/7$};
        \node[draw, circle] (x1) at (1.75,2.75) {$x_1$};
        \node (vx1) at (2,3.3) {$4/7$};
        \node[draw, circle] (x2) at (1,0) {$x_2$};
        \node (vx2) at (0.25,0) {$2/7$};
        \node[draw, diamond, aspect=1] (r2) at (2,4.5) {$r_2$};
        \node (vr2) at (2.5,5) {$1/7$};
        \node[draw, circle] (x3) at (2.5,1) {$x_3$};
        \node (vx3) at (2.5,0.4) {$0$};
        \node[draw, diamond, aspect=1] (r3) at (3.5,2) {$r_3$};
        \node (vr3) at (4,2.5) {$4/7$};
        \node[draw, circle] (x4) at (4.5,0) {$x_4$};
        \node (vx4) at (5.1,0) {$0$};
        \node[draw, circle] (x5) at (4,4) {$x_5$};
        \node (vx5) at (4.75,4) {$1/7$};
        \node[draw, circle, fill=gray!50] (1) at (5,2) {$1$};
        
		\node (sig1) at (0.75,4.75) {$\sigma$:};
		\node (sig2) at (8.75,4.75) {$\sigma'$:};
		\node (sig3) at (4.75,-1.25) {$\sigma''$:};
		
		\node (tr) at (6.5,2.25) {$\longrightarrow$};
		
		\node[draw, circle, fill=gray!50] (0b) at (8,4) {$0$};
        \node[draw, diamond, aspect=1] (r1b) at (8,2) {$r_1$};
        \node (vr1b) at (7.5,2.5) {$2/7$};
        \node[draw, circle] (x1b) at (9.75,2.75) {$x_1$};
        \node (vx1b) at (10,3.3) {$1/7$};
        \node[draw, circle] (x2b) at (9,0) {$x_2$};
        \node (vx2b) at (8.25,0) {$2/7$};
        \node[draw, diamond, aspect=1] (r2b) at (10,4.5) {$r_2$};
        \node (vr2b) at (10.5,5) {$1/7$};
        \node[draw, circle] (x3b) at (10.5,1) {$x_3$};
        \node (vx3b) at (10.5,0.4) {$4/7$};
        \node[draw, diamond, aspect=1] (r3b) at (11.5,2) {$r_3$};
        \node (vr3b) at (12,2.5) {$4/7$};
        \node[draw, circle] (x4b) at (12.5,0) {$x_4$};
        \node (vx4b) at (13.1,0) {$1$};
        \node[draw, circle] (x5b) at (12,4) {$x_5$};
        \node (vx5b) at (12.75,4) {$1/7$};
        \node[draw, circle, fill=gray!50] (1b) at (13,2) {$1$};
        
        \node (tr) at (2.5,-3.75) {$\longrightarrow$};
        
        \node[draw, circle, fill=gray!50] (0c) at (4,-2) {$0$};
        \node[draw, diamond, aspect=1] (r1c) at (4,-4) {$r_1$};
        \node (vr1c) at (3.5,-3.5) {$2/7$};
        \node[draw, circle] (x1c) at (5.75,-3.25) {$x_1$};
        \node (vx1c) at (6,3.3) {$1/7$};
        \node[draw, circle] (x2c) at (5,-6) {$x_2$};
        \node (vx2c) at (4.25,-6) {$4/7$};
        \node[draw, diamond, aspect=1] (r2c) at (6,-1.5) {$r_2$};
        \node (vr2c) at (6.5,-1) {$1/7$};
        \node[draw, circle] (x3c) at (6.5,-5) {$x_3$};
        \node (vx3c) at (6.5,-5.6) {$4/7$};
        \node[draw, diamond, aspect=1] (r3c) at (7.5,-4) {$r_3$};
        \node (vr3c) at (8,-3.5) {$4/7$};
        \node[draw, circle] (x4c) at (8.5,-6) {$x_4$};
        \node (vx4c) at (9.1,-6) {$1$};
        \node[draw, circle] (x5c) at (8,-2) {$x_5$};
        \node (vx5c) at (8.75,-2) {$1/7$};
        \node[draw, circle, fill=gray!50] (1c) at (9,-4) {$1$};

		\draw[->,>=latex, dashed] (x1)--(0);
		\draw[->,>=latex] (x1)--(r3);
		\draw[->,>=latex] (x2)--(r1);
		\draw[->,>=latex, dashed] (x2)--(x3);
		\draw[->,>=latex, dashed] (x3)--(r3);
		\draw[->,>=latex] (x3) to[bend right = 20] (x4);
		\draw[->,>=latex] (x4) to[bend right = 20] (x3);
		\draw[->,>=latex, dashed] (x4)--(1);
		\draw[->,>=latex] (x5)--(r2);
		\draw[->,>=latex, dashed] (x5)--(1);
		
		\draw[->,>=latex] (r1)--(0);
		\draw[->,>=latex] (r1)--(r3);
		\draw[->,>=latex] (r2)--(0);
		\draw[->,>=latex] (r2)--(r1);
		\draw[->,>=latex] (r3)--(r2);
		\draw[->,>=latex] (r3)--(1);
		
		\draw[->,>=latex, dashed] (x1b)--(0b);
		\draw[->,>=latex] (x1b)--(r3b);
		\draw[->,>=latex] (x2b)--(r1b);
		\draw[->,>=latex, dashed] (x2b)--(x3b);
		\draw[->,>=latex] (x3b)--(r3b);
		\draw[->,>=latex, dashed] (x3b) to[bend right = 20] (x4b);
		\draw[->,>=latex, dashed] (x4b) to[bend right = 20] (x3b);
		\draw[->,>=latex] (x4b)--(1b);
		\draw[->,>=latex] (x5b)--(r2b);
		\draw[->,>=latex, dashed] (x5b)--(1b);
		
		\draw[->,>=latex] (r1b)--(0b);
		\draw[->,>=latex] (r1b)--(r3b);
		\draw[->,>=latex] (r2b)--(0b);
		\draw[->,>=latex] (r2b)--(r1b);
		\draw[->,>=latex] (r3b)--(r2b);
		\draw[->,>=latex] (r3b)--(1b);

		\draw[->,>=latex, dashed] (x1c)--(0c);
		\draw[->,>=latex] (x1c)--(r3c);
		\draw[->,>=latex, dashed] (x2c)--(r1c);
		\draw[->,>=latex] (x2c)--(x3c);
		\draw[->,>=latex] (x3c)--(r3c);
		\draw[->,>=latex, dashed] (x3c) to[bend right = 20] (x4c);
		\draw[->,>=latex, dashed] (x4c) to[bend right = 20] (x3c);
		\draw[->,>=latex] (x4c)--(1c);
		\draw[->,>=latex] (x5c)--(r2c);
		\draw[->,>=latex, dashed] (x5c)--(1c);
		
		\draw[->,>=latex] (r1c)--(0c);
		\draw[->,>=latex] (r1c)--(r3c);
		\draw[->,>=latex] (r2c)--(0c);
		\draw[->,>=latex] (r2c)--(r1c);
		\draw[->,>=latex] (r3c)--(r2c);
		\draw[->,>=latex] (r3c)--(1c);

    \end{tikzpicture}

\caption{The strategy of \M\ are represented by plain arcs and the probability distribution on the random vertices is the uniform distribution. The switch set of $\sigma$, $S_{\sigma}$ is $\{x_3,x_4,x_5\}$. The strategy $\sigma'$ is a $\sigma$-switch and $\sigma''$ is a $(\sigma,S_{\sigma})$-super-switch.}
\label{fig:def_ss}
\end{figure}
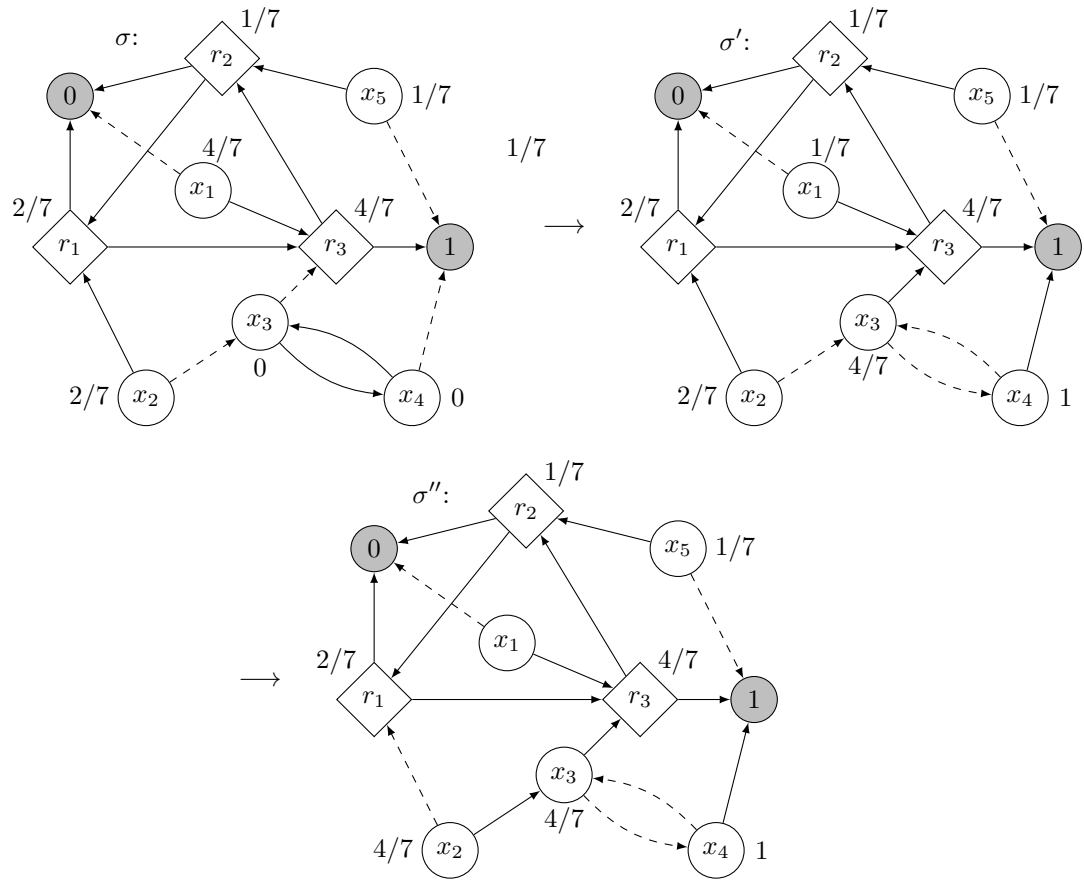
        
        \begin{lemma}
            For $\sigma$ a non-optimal \M\ strategy, $T$ a set of \M\ vertices and $\sigma'$ a $(\sigma,T)$-super-switch, $v_{\sigma'} > v_{\sigma}$.
        \end{lemma}
        
        \begin{proof}
            We consider the strategy $\sigma''$ that plays as $\sigma'$ on T and as $\sigma$ on $V_{\M} \smallsetminus T$. The strategy $\sigma''$ is a $\sigma$-switch and by Proposition~\ref{prop:switch} $v_{\sigma''} > v_{\sigma}$. Moreover, $\sigma'$ is the optimal strategy of $G_{|T[\sigma'']}$ and $\sigma''$ is also a strategy of $G_{|T[\sigma'']}$. Hence, $v_{\sigma'} \geq v_{\sigma''}$ and $v_{\sigma'} > v_{\sigma}$.
        \end{proof}
        
        In~\cite{tripathi2011strategy}, Tripathi, Valkanova and Kumar, show that for a strategy $\sigma$, there is at least $|S|-1$ $\sigma$-switch $\sigma'$ different from $\bar{\sigma}$ such that $v_{\sigma'} \leq v_{\bar{\sigma}}$. In Theorem~\ref{thm:tripathi_super_switch} we adapt their proof to show that this result can be extended to super-switch.
        
        For $\sigma$ a \M\ strategy, $T$ a set of \M\ vertices and $\sigma'$ a $(\sigma,T)$-super-switch, we note $D_{\sigma,\sigma'}$ the set $\{x\;|\;\sigma(x) \neq \sigma'(x), x \in T\}$ the set of vertices that has been switched in the intermediate $\sigma$-switch. 
        
        \begin{theorem}\label{thm:tripathi_super_switch}
            For $\sigma$ a non-optimal \M\ strategy, $T$ a set of \M\ vertices intersecting $S_{\sigma}$ and $\sigma'$ a $(\sigma,T)$-super-switch, there is at least $|D_{\sigma,\sigma'}|-1$ $(\sigma,T)$-super-switches $\sigma'' \neq \sigma'$ such that $v_{\sigma} < v_{\sigma''} \leq v_{\sigma'}$.
        \end{theorem}
        
        \begin{proof}
            First, we show that Theorem~\ref{thm:tripathi_super_switch} is true in the case $D_{\sigma,\sigma'} = 2$. Let $D_{\sigma,\sigma'} = \{x;y\}$. We consider the game $G'$ which is the game $G_{|T \smallsetminus \{x;y\}[\sigma]}$ where all edges $(x,x')$ and $(y,y')$ with $x' \notin \{\sigma(x);\sigma'(x)\}$ and $y' \notin \{\sigma(y),\sigma'(y)\}$ has been removed. Both $\sigma$ and $\sigma'$ are strategies in the game $G'$. We write $S_{\tilde{\sigma}}$ the switch set of a strategy $\tilde{\sigma}$ in the game $G$ and $S'_{\tilde{\sigma}}$ the switch set of a strategy $\tilde{\sigma}$ in the game $G'$. Hence, $S'_{\sigma} = \{x;y\}$. Let call $\sigma_{x}$ and $\sigma_{y}$ the $(\sigma,T)$-super-switch where respectively only $x$ and only $y$ has been switched in $T$.
            
            By Lemma~\ref{prop:set_inclusion}, $S'_{\sigma'}$ is strictly included in $\{x;y\}$ and thus is at most a singleton. If $S'_{\sigma'}$ is empty then,  $v_{\sigma} < v_{\sigma_x}, v_{\sigma_y} \leq v_{\sigma'}$. We suppose that $S'_{\sigma'} = \{x\}$. Then, we notice that $\sigma_x$ and $\sigma'$ are both strategies of $G'_{|\{x\}[\sigma']}$ and $\sigma'$ is optimal in $G'_{|\{x\}[\sigma']}$. Thus, $v_{\sigma} < v_{\sigma_x} \leq v_{\sigma'}$ and $v_{\sigma_y} > v_{\sigma'}$. This implies that if $v_{\sigma_x} \ngtr v_{\sigma_y}$, $v_{\sigma'} \geq v_{\sigma_x}$ and if we also have $v_{\sigma_x} \neq v_{\sigma_y}$ then $y \in S'_{\sigma_x}$
            
            Now we look at the general case. Let $D_{\sigma,\sigma'} =\{x_,\ldots,x_t\}$. For $E \subset \{1,\ldots,t\}$, we write $v_{E}$ the value of the $(\sigma,T)$-super-switch $\sigma_{\{E\}}$ where only the vertices $x_i$ for $i \in E$ has been switched. We assume that for all $i < j \leq t$, $v_{\{i\}} \ngtr v_{\{j\}}$. If for all $j > 1$, $v_{\{1\}} \neq v_{\{j\}}$ then $v_{\{1;j\}} \geq v_{\{1\}}$ and $j \in S_\{\sigma_{\{1\}}\}$. 
            
            Otherwise, we suppose that there is $k > 1$ such that for all $2 \leq j \leq k$, $v_{\{1\}} = v_{\{j\}}$ and for all $j>k$, $v_{\{1\}} \neq v_{\{j\}}$. We consider the game $G'$ which is the game $G_{|T \smallsetminus \{x_1;\ldots;x_k\}[\sigma]}$ and where for $i \leq k$, all edges from $x_i$ not towards $\sigma(x_i)$ or $\sigma'(x_i)$ are removed. We note $S'$ the switch sets in $G'$ and $\sigma' = \sigma_{\{1;\ldots;k\}}$. By Lemma~\ref{prop:set_inclusion}, $S'_{\sigma'}$ is strictly included in $\{1;\ldots;k\}$. We suppose that $S'_{\sigma'} = \{1;\ldots;k'\}$ for some $k' < k$. By induction hypothesis, we know that $v_{\{1\}} \leq v_{\{1;\ldots;k'\}}$ and $\sigma'$ and $\sigma_{\{1;\ldots;k'\}}$ are both optimal strategies of $G'_{|\{1;\ldots;k'\}[\sigma']}$ and thus have same value. Thus, $v_{E} \leq v_{\sigma'}$ for all $E \subseteq \{1;\ldots;k\}$ and for all $j > k$, $x_j \in S_{\sigma'}$. We conclude by induction on the $\{x_j;\ldots;x_t\}$.
            
            Thus, we have that $v_{\sigma'} = v_{\{1,2,\ldots,t\}} \geq v_{\{1,2,\ldots,t-1\}} \geq \ldots \geq v_{\{1\}}$ and we have proven that there is at least $|D_{\sigma,\sigma'}|-1$ $(\sigma,T)$-super-switchs $\sigma'' \neq \sigma'$ such that $v_{\sigma} < v_{\sigma''} \leq v_{\sigma'}$.
        \end{proof}
        
        We notice that a $\sigma$-switch is a $(\sigma,V_{\M})$-super-switch. Hence, Theorem~\ref{thm:tripathi_super_switch} also proves the main theorem of~\cite{tripathi2011strategy}.

\section{A recursive algorithm with a pair of fixed vertices}\label{sec:rec2}
    
        In all this section, we only consider SSGs of degree $d$.
    
        Algorithm~\ref{alg:RP} works by fixing the strategy of the game on two vertices, then recursively solving the rest of the game. If this does not yield an optimal strategy, then it switches the strategy on the fixed vertices and iterate. The switch sets of the considered strategies after the recursive call (line $9$)  are included in of $\{x,y\}$.

        \begin{algorithm}
        	\caption{RecursivePair\label{alg:RP}}
        	\DontPrintSemicolon
        	\KwData{$G$ an SSG}
        	\KwResult{$\sigma$ an optimal \M\ strategy and $v$ the optimal value vector.}
        	\Begin{
        	    \If{$|V_{\M}| \leq 1$}{
        	        Compute the optimal strategy $\sigma$ by testing all possibilities\;
        	        \KwRet{$(\sigma,v_{\sigma})$}
        	    }
        	    $\sigma \longleftarrow $ a \M\ strategy\;
        	    $x,\;y \longleftarrow $ two vertices of $V_{\M}$\;
        	    $(\sigma,v) \longleftarrow $ RecursivePair($G_{|\{x,y\}[\sigma]}$)\;
        	    \While{$\sigma$ is not optimal}{
        	        $\sigma \longleftarrow \bar{\sigma}$\;
        	        $(\sigma,v) \longleftarrow $ RecursivePair($G_{|\{x,y\}[\sigma]}$)\;
        	    }
        	    \KwRet{$(\sigma,v)$}
        	}
        \end{algorithm}
        
        Let us first recall that computing $v_{\sigma}$ can be done in polynomial time in $|G|$ by solving a linear programming problem.
        
        We write $N$ the number of iterations of the loop. Let $\sigma_i$ be the value of $\sigma$ at the start of the $i$-th iteration of the loop line~8 and $\sigma_{N+1}$ the value of $\sigma$ after the last iteration of the loop. Algorithm~\ref{alg:RP} makes $N+1$ recursive calls to an instance with $n-2$ \M\ vertices. We notice that for all $i$, $S_{\sigma_i} \subseteq \{x;y\}$.
        
        \begin{lemma}\label{lem:d1}
            For all $i<N+1$, $|S_{\sigma_i}| \neq 0$. Moreover, tere is at most $2(d-1)$ indices $k$ such that $|S_{\sigma_k}|=1$. 
        \end{lemma}
        
        \begin{proof}
            If $S_{\sigma_i} = \emptyset$, then the algorithm stops and $i=N+1$. Thus, for all $k < N+1$, $S_{\sigma_k} \neq \emptyset$.
            
            For all neighbours $x'$ of $x$ there is at most one $k$ such that $\sigma_k(x) = x'$ and $S_{\sigma_k}= \{x\}$ since such strategies are all optimal in $G_{|\{x\}[\sigma_k]}$ and thus have the same value. Moreover, if there is an optimal strategy $\sigma^{*}$ such that $\sigma^*(x) = x'$, then all optimal strategies of $G_{|\{x\}[\sigma^{*}]}$ are optimal on $G$ and there is no strategy $\sigma$ such that $\sigma(x) = x'$ and $S_{\sigma}=\{x\}$. Hence, there is at most $2(d-1)$ visited strategies with switch set of size $1$.
        \end{proof}
        
        \begin{proposition}
            Algorithm~\ref{alg:RP} runs in $O\left(\left(\left\lfloor\frac{(d+1)^2}{2}\right\rfloor-1\right)^{n/2}Poly(|G|)\right)$.
        \end{proposition}
        
        \begin{proof}
            If we write, $n_0$, $n_1$ and $n_2$ the number of indices $i$ such that $\sigma_i$ has a switch set of respectively size $0$, $1$ and $2$, then $N+1 = n_0 + n_1 + n_2$. By Theorem~\ref{thm:tripathi_super_switch}, if $S_{\sigma_i}=\{x;y\}$, then there is a super-switch $\sigma'$ such that $v_{\sigma_i} < v_{\sigma'} \leq v_{\sigma_{i+1}}$. Thus, $n_0 + n_1 + 2n_2 \leq d^2$. We also know by Lemma~\ref{lem:d1} that $n_0 + n_1 \leq 2d-1$. Then:
            \[2(N+1) = (n_0 + n_1) + (n_0+n_1+2n_2) \leq d^2 + 2d - 1 \]
            Which gives:
            \[N+1 \leq \frac{(d+1)^{2}}{2} - 1\]
            Since $N+1$ is an integer, we have $N+1 \leq \left\lfloor\frac{(d+1)^2}{2}-1\right\rfloor$. Hence, Algorithm~\ref{alg:RP} makes at most $\left\lfloor\frac{(d+1)^2}{2}-1\right\rfloor$ recursive calls to an instance with $n-2$ \M\ vertices and Algorithm~\ref{alg:RP} runs in $O\left(\left(\left\lfloor\frac{(d+1)^2}{2}\right\rfloor-1\right)^{n/2}Poly(|G|)\right)$.
        \end{proof}

        In the case of binary SSG, Algorithm~\ref{alg:RP} is similar to Ludwig's Algorithm~\cite{ludwig1995subexponential} which fixes the strategy on the vertices one at a time. The choice of which vertex to fix is random and provides an algorithm that runs in expected time $2^{O\left(\sqrt{n}\right)}Poly\left(|V|\right)$. However, despite the proximity of the two algorithms, we were yet not able to find a similar analysis as the one in~\cite{ludwig1995subexponential} to the stochastic version of Algorithm~\ref{alg:RP}.
        
        On binary SSG, Algorithm~\ref{alg:RP} gives a complexity bound in $O\left(\sqrt{3}^nPoly(|G|)\right)$ which is better that the currently known one for binary SSG in~\cite{tripathi2011strategy}. However, it is still possible to improve this complexity, as we show in the next section.
        
\section{A Recursive Algorithm for Binary SSGs}\label{sec:recT}
    
        In all this section, we will only consider binary SSG.
    
        The concept of Algorithm~\ref{alg:RP} is to fix a subset of vertices and recursively solve the rest of the game. Then, we switch the current strategy and fix a smaller subset of vertices and reiterate. We show that we never make a call to an instance with $n-1$ \M\ vertices. This is done by carefully selecting the set of fixed vertices.
        
        \begin{algorithm}
        	\caption{DecreasingFixedSet\label{alg:RT}}
        	\DontPrintSemicolon
        	\KwData{$G$ an SSG}
        	\KwResult{$\sigma$ an optimal \M\ strategy and $v$ the optimal value vector.}
        	\Begin{
        	    $\sigma \longleftarrow $ a \M\ strategy\;
        	    $S \longleftarrow S_{\sigma}$\;
        	    $\sigma \longleftarrow \bar{\sigma}$\;
        	    $v \longleftarrow v_{\sigma}$\;
        	    $T \longleftarrow S \cup S_{\sigma}$\;
        	    $S \longleftarrow S_{\sigma}$\;
        	    \While{$S \neq \emptyset$}{
        	        $\sigma \longleftarrow \bar{\sigma}$\;
        	        $(\sigma,v) \longleftarrow $ DecreasingFixedSet($G_{|T[\sigma]}$)\;
        	        $T \longleftarrow S \cup S_{\sigma}$\;
        	        $S \longleftarrow S_{\sigma}$\;
        	    }
        	    \KwRet{$(\sigma,v)$}
        	}
        \end{algorithm}
        
        As stated before, the goal of Algorithm~\ref{alg:RT} is to avoid the call to a game with $n-1$ \M\ vertices. In order to achieve this, the set of vertices that is fixed in the recursive call is the union of the previous and current switch set. Computing $S_{\sigma}$ at line 11 can be done in linear time with the value vector $v$ computed at line 10. 
        
        \begin{lemma}
            Algorithm~\ref{alg:RT} terminates and computes an optimal \M\ strategy and its value vector.
        \end{lemma}
        
        \begin{proof}
            The value vectors of the visited strategies are increasing by Proposition~\ref{prop:switch} and there are a finite number of \M\ strategies, hence Algorithm~\ref{alg:RT} terminates. The algorithm ends when the switch set of a \M\ strategy is empty, hence when the algorithm terminates it computes an optimal \M\ strategy. Alternatively, Algorithm~\ref{alg:RT} is a strategy improvement algorithm, thus by~\cite{auger2021}, it terminates and computes an optimal strategy.
        \end{proof}
        
        Let us call $T_0 = V_{\M}$ and $S_0$ and $\sigma_0$ the value of the variables $S$ and $\sigma$ after line 3. In addition, we call $\sigma_i$, $T_i$ and $S_i$ the value of the variables $\sigma$, $T$ and $S$ at the beginning of the $i$-th iteration of the while loop. Let $N$ be the number of iterations. We call $T_{N+1}$, $S_{N+1}$ and $\sigma_{N+1}$ the value of those variables at the end of the last iteration. By line 11 of Algorithm~\ref{alg:RT} for every $i$, $S_i \subset T_i$. We create a partition of $T_i$ by considering $S_i$ and $S_{i-1} \smallsetminus S_{i}$.
        
        If for all $i \geq 1$, $S_{i-1} \smallsetminus S_{i}$ is not empty, then $|T_i| > |S_i|$ and $S_i$ not empty implies that $|T_i| \geq 2$. Then, all recursive calls to Algorithm~\ref{alg:RT} are made to a subgame with at most $n-2$ \M\ vertices.
        
        \begin{proposition}\label{prop:u_not_empty}
            For every $1 \leq i \leq N$, $S_{i-1} \smallsetminus S_{i}$ is not empty.
        \end{proposition}
        
        \begin{proof}
            For every $1 \leq i \leq N$, by definition $T_i = S_{i-1} \cup S_{i}$. The strategy $\sigma_{i}$ is a $(\sigma_{i-1},T_{i-1})$-super-switch, thus $v_{\sigma_i} > v_{\sigma_{i-1}}$ and $S_{i-1}$ is not a subset or equal to $S_{i}$ according to the contraposition of Proposition~\ref{prop:set_inclusion}.
        \end{proof}
        
        Now, we need to prove that each iteration of the loop strictly decreases the size of $T$.
        
        \begin{proposition}\label{prop:t_decreasing}
            For $1 \leq i < N$, $T_{i+1} \subsetneq T_{i}$.
        \end{proposition}
        
        \begin{proof}
            Let $1 \leq i < N$.
            First of all, we notice that $S_i \subseteq T_{i-1}$ since $\sigma_i$ is optimal in $G_{T_{i-1}[\sigma_i]}$, and $T_i = S_{i-1} \cup S_i$. Thus, we have $T_i \subseteq T_{i-1}$.
            
            We recall that strategy $\sigma_{i-1}$ is optimal in the game $G_{S_{i-1}[\sigma_{i-1}]}$. We notice that for every $x \in S_{i-1} \cap S_i$, $\sigma_{i+1}(x) = \sigma_{i-1}(x)$: the strategy of every vertex in $S_{i-1} \cap S_i$ has been changed twice, thus going back to its original value since we consider binary SSG. We recall that $S_{i-1} \smallsetminus S_{i}$ is not empty by Proposition~\ref{prop:u_not_empty} and assume for the sake of contradiction that $S_{i-1} \smallsetminus S_{i} \subset S_{i+1}$. Then we define the strategy $\sigma'$ as follows:
            \[\forall x \in S_{i-1} \smallsetminus S_{i},\; \sigma'(x) \neq \sigma_{i+1}(x)\]
            \[\forall x \notin S_{i-1} \smallsetminus S_{i},\; \sigma'(x) = \sigma_{i+1}(x)\]
            Since $S_{i} \smallsetminus S_{i-1}$ is a subset of $S_{i+1}$, $\sigma'$ is a $\sigma_{i+1}$-switch and by Proposition~\ref{prop:switch} $\sigma' > \sigma_{i+1} > \sigma_{i}$. However, for all $x \in S_{i-1}$, $\sigma'(x) = \sigma_{i-1}(x)$ and $\sigma'$ is a strategy of $G_{S_{i-1}[\sigma_{i-1}]}$ which contradicts the optimality of $\sigma_{i-1}$ on this game. This shows that there exists $x$ in $S_{i-1} \smallsetminus S_{i}$ but not in $S_{i+1}$. In other words, there is $x$ in $S_{i-1}$ and thus in $T_i$ but not in $S_{i} \cup S_{i+1}$ and thus not in $T_{i+1}$. Therefore, we have proven that $T_{i+1} \subsetneq T_{i}$. In order to better visualise this proof a representation of the successive switches is provided Figure~\ref{fig:double_switch}.
        \end{proof}
        
        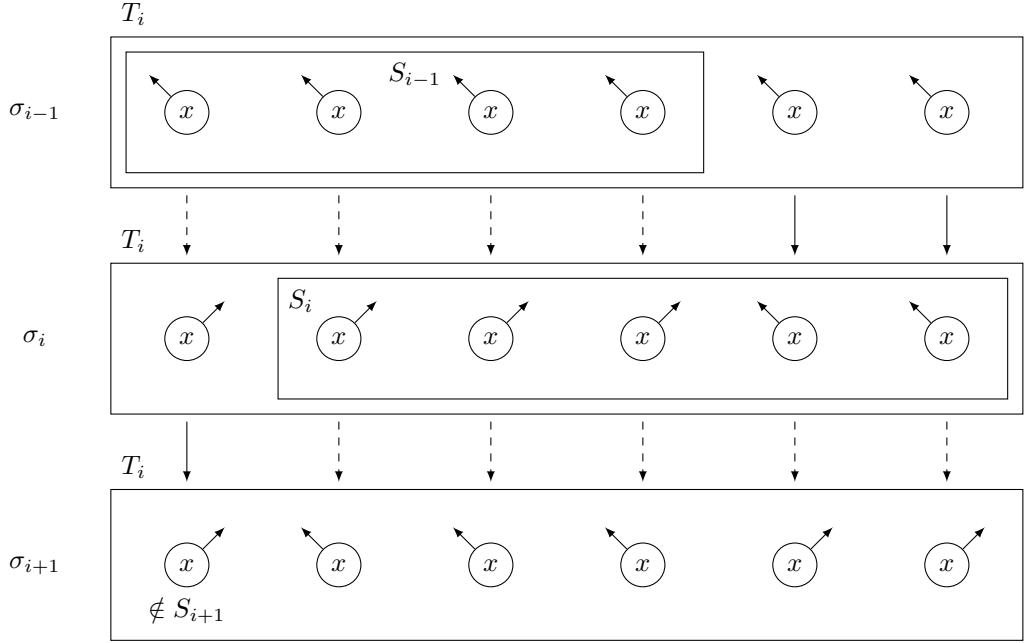
\begin{figure}
\centering

    \begin{tikzpicture}
    
        \draw (0,8) rectangle (12,10);
        \draw (0.2,8.2) rectangle (7.8,9.8);
        
        \node[draw, circle] (x1) at (1,9) {$x$};
        \node[draw, circle] (x2) at (3,9) {$x$};
        \node[draw, circle] (x3) at (5,9) {$x$};
        \node[draw, circle] (x4) at (7,9) {$x$};
        \node[draw, circle] (x5) at (9,9) {$x$};
        \node[draw, circle] (x6) at (11,9) {$x$};
        
        \node (s1) at (-1,9) {$\sigma_{i-1}$};
        \node (S1) at (4,9.5) {$S_{i-1}$};
        
        \node (T1) at (0.3,10.3) {$T_{i}$};
        
        \draw (0,5) rectangle (12,7);
        \draw (2.2,5.2) rectangle (11.8,6.8);
        
        \node[draw, circle] (x1b) at (1,6) {$x$};
        \node[draw, circle] (x2b) at (3,6) {$x$};
        \node[draw, circle] (x3b) at (5,6) {$x$};
        \node[draw, circle] (x4b) at (7,6) {$x$};
        \node[draw, circle] (x5b) at (9,6) {$x$};
        \node[draw, circle] (x6b) at (11,6) {$x$};
        
        \node (s2) at (-1,6) {$\sigma_{i}$};
        \node (S2) at (2.5,6.5) {$S_{i}$};
        \node (T2) at (0.3,7.3) {$T_{i}$};
        
        \draw (0,2) rectangle (12,4);
        
        \node[draw, circle] (x1c) at (1,3) {$x$};
        \node[draw, circle] (x2c) at (3,3) {$x$};
        \node[draw, circle] (x3c) at (5,3) {$x$};
        \node[draw, circle] (x4c) at (7,3) {$x$};
        \node[draw, circle] (x5c) at (9,3) {$x$};
        \node[draw, circle] (x6c) at (11,3) {$x$};
        
        \node (s3) at (-1,3) {$\sigma_{i+1}$};
        \node (nS3) at (1,2.4) {$\notin S_{i+1}$};
        
        \node (T3) at (0.3,4.3) {$T_{i}$};
        
        \draw[->,>=latex,dashed] (1,7.9)--(1,7.1);
        \draw[->,>=latex] (1,4.9)--(1,4.1);
        \draw[->,>=latex,dashed] (3,7.9)--(3,7.1);
        \draw[->,>=latex,dashed] (3,4.9)--(3,4.1);
        \draw[->,>=latex,dashed] (5,7.9)--(5,7.1);
        \draw[->,>=latex,dashed] (5,4.9)--(5,4.1);
        \draw[->,>=latex,dashed] (7,7.9)--(7,7.1);
        \draw[->,>=latex,dashed] (7,4.9)--(7,4.1);
        \draw[->,>=latex] (9,7.9)--(9,7.1);
        \draw[->,>=latex,dashed] (9,4.9)--(9,4.1);
        \draw[->,>=latex] (11,7.9)--(11,7.1);
        \draw[->,>=latex,dashed] (11,4.9)--(11,4.1);
        
        \draw[->,>=latex] (x1)--(0.5,9.5);
        \draw[->,>=latex] (x2)--(2.5,9.5);
        \draw[->,>=latex] (x3)--(4.5,9.5);
        \draw[->,>=latex] (x4)--(6.5,9.5);
        \draw[->,>=latex] (x5)--(8.5,9.5);
        \draw[->,>=latex] (x6)--(10.5,9.5);
        
        \draw[->,>=latex] (x1b)--(1.5,6.5);
        \draw[->,>=latex] (x2b)--(3.5,6.5);
        \draw[->,>=latex] (x3b)--(5.5,6.5);
        \draw[->,>=latex] (x4b)--(7.5,6.5);
        \draw[->,>=latex] (x5b)--(8.5,6.5);
        \draw[->,>=latex] (x6b)--(10.5,6.5);
        
        \draw[->,>=latex] (x1c)--(1.5,3.5);
        \draw[->,>=latex] (x2c)--(2.5,3.5);
        \draw[->,>=latex] (x3c)--(4.5,3.5);
        \draw[->,>=latex] (x4c)--(6.5,3.5);
        \draw[->,>=latex] (x5c)--(9.5,3.5);
        \draw[->,>=latex] (x6c)--(11.5,3.5);

    \end{tikzpicture}

\caption{Strategy on the vertices of $T_i$ under strategies $\sigma_{i-1}$, $\sigma_{i}$ and $\sigma_{i+1}$}
\label{fig:double_switch}
\end{figure}
        
        We can now give a bound on the complexity of Algorithm~\ref{alg:RT}.
        
        \begin{theorem}
            Algorithm~\ref{alg:RT} has time complexity $O\left(\varphi^{n}Poly(|G|)\right)$, where $\varphi = \frac{1+\sqrt{5}}{2}$ is the golden ratio.
        \end{theorem}
        
        \begin{proof}
            We denote by $C(0,k)$, the complexity of solving an SSG with $0$ \M\ vertices and $k$ total vertices. This is the resolution of a one-player game and can be done in polynomial time in the size of the game by solving a linear programming problem. We define $C(n,k)$ as:
            \[C(n,k) = C(n-2,k) + C(n-3,k) + \ldots + C(1,k) + 3C(0,k)\]
            We show by induction that the complexity of solving an SSG using Algorithm~\ref{alg:RT} with $n$ \M\ vertices and $k$ total vertices is bounded by $C(n,k)$. According to Proposition~\ref{prop:t_decreasing}, each call to DecreasingFixedSet is done on an SSG with a decreasing number of \M\ vertices. Proposition~\ref{prop:u_not_empty} also stipulates that if $S_{i}$ is not empty, then $S_{i-1} \smallsetminus S_{i}$ is also not empty and $|T_{i}|$ is greater than one. Thus, each recursive call is made on an instance with at most $n-2$ \M\ vertices. Finally, $v_{\sigma}$ is computed twice before the loop, costing $C(0,k)$ operations. Therefore Algorithm~\ref{alg:RT} has time complexity $O\left(C(n,k)\right)$.
            We notice that $C(n,k) - C(n-1,k) = C(n-2,k)$. Thus, we have:
            \[C(n,k) = O\left(\varphi^{n}Poly(|G|)\right)\]
            
            Thus, we have shown that Algorithm~\ref{alg:RT} has time complexity $O\left(\varphi^{n}Poly(|G|)\right)$.
        \end{proof}
        
        The polynomial factor in all our Algorithms corresponds to the complexity of computing $v_{\sigma}$ from $\sigma$. We recall that this is the complexity of solving a linear programming problem with $|V|$ variables. It is the same polynomial factor as the one in Tripathi, Valkanova and Kumar's algorithm~\cite{tripathi2011strategy} which runs in $O\left(2/n \cdot Poly(|G|)\right)$.
        
        However, the analysis of Algorithm~\ref{alg:RT} does not hold in the case of SSG with higher degree. Algorithm~\ref{alg:RP} can still be improved for some degree by changing the size of the fixed set according to $d$. For instance, if we fix set of size $3$ the complexity of solving SSG of degree $3$ is $O\left(17^{n/3}\right)$ iterations instead of $O\left(7^{n/2}\right)$. For information, $17^{1/3} \simeq 2.57$ and $7^{1/2} \simeq 2.65$. However, increasing the size of the fixed set not always hold better complexity. For binary SSG, the number of iterations with set of size $2$ is $O\left(3^{n/2}\right)$ and $O\left(6^{n/3}\right)$ for set of size $3$ and we know that $6^{1/3} \simeq 1.82$ and $3^{1/2} \simeq 1.73$.

\bibliography{algoseq.bib}

\appendix

\end{document}